\documentclass[12pt,reqno]{amsart}

\newcommand\version{January 18, 2012}

\usepackage{amsmath,amsfonts,amsthm,amssymb,amsxtra}

\newtheorem{theorem}{Theorem}
\newtheorem{proposition}[theorem]{Proposition}
\newtheorem{lemma}[theorem]{Lemma}

\theoremstyle{definition}

\theoremstyle{remark}

\numberwithin{equation}{section}

\newcommand{\C}{\mathbb{C}}

\renewcommand{\epsilon}{\varepsilon}

\renewcommand{\phi}{\varphi}
\newcommand{\R}{\mathbb{R}}

\newcommand{\Sph}{\mathbb{S}}

\DeclareMathOperator{\ran}{ran}
\DeclareMathOperator{\re}{Re}
\DeclareMathOperator{\spa}{span}

\begin{document}

\title[Symmetry of bipolaron bound states --- \version]{Symmetry of bipolaron bound states\\ for small Coulomb repulsion}

\author[R. L. Frank]{Rupert L. Frank}
\address{Rupert L. Frank, Department of Mathematics, Princeton University, Princeton, NJ 08544, USA}
\email{rlfrank@math.princeton.edu}

\author[E. H. Lieb]{Elliott H. Lieb}
 \address{E.H. Lieb, Departments of Mathematics and Physics, Princeton University, P.O. Box 708, Princeton, NJ 08544, USA} \email{lieb@princeton.edu}

\author[R. Seiringer]{Robert Seiringer}
 \address{R. Seiringer, Department of Mathematics and Statistics, McGill University, 805 Sherbrooke Street West, Montreal, QC H3A 2K6, Canada} 
\email{rseiring@math.mcgill.ca}

\begin{abstract}
We consider the bipolaron in the Pekar--Tomasevich approximation and address the question whether the ground state is spherically symmetric or not. Numerical analysis has, so far, not completely settled the question. Our contribution is to prove rigorously that the ground state remains spherical for small values of the electron-electron Coulomb repulsion.
\end{abstract}


\maketitle

\renewcommand{\thefootnote}{${}$} \footnotetext{\copyright\, 2012 by
  the authors. This paper may be reproduced, in its entirety, for
  non-commercial purposes.}

\section{Introduction}

In this paper we shall be concerned with properties of the bound state of two polarons. We do this in the context of the Pekar--Tomasevich model \cite{Pe,PeTo} of the large polaron, which, in turn, is based on Fr\"ohlich's polaron model \cite{Fr}. In the latter model two electrons interact with a quantized electric field generated by the displacement of the nuclei in a polar lattice. There are two coupling constants in Fr\"ohlich's model. The coupling to the field, $\alpha$, and the Coulomb repulsion among the electrons, denoted by $U$. Pekar's approximation is to assert that the wave function is a product of a two-particle electron wave function $\psi$ times a field function $\Phi$. After eliminating the field one is led to Pekar's energy expression for $\psi$,
\begin{equation}
 \label{eq:func}
\mathcal E_U[\psi] = \iint_{\R^3\times\R^3} \left( |\nabla_x\psi|^2 +|\nabla_y\psi|^2 + \frac{U}{|x-y|} |\psi|^2 \right) \,dx\,dy - 2 \alpha D[\rho_\psi,\rho_\psi] \,.
\end{equation}
The electron coordinates are $x$ and $y$ and the electron spin does not appear explicitly, except that $\psi$ is symmetric for the ground state, which is a singlet state. (The reason that it is a singlet is that the ground state is a positive function and must, therefore, be symmetric -- an observation that goes back to Wigner many years ago.)

In \eqref{eq:func} $\rho_\psi$ denotes the electron density, given by
$$
\rho_\psi(x) = \int_{\R^3} |\psi(x,y)|^2 \,dy + \int_{\R^3} |\psi(y,x)|^2 \,dy \,,
$$
and $D[\rho,\rho]$ is the Coulomb energy of a charge distribution $\rho$,
$$
D[\rho,\rho] = \frac12 \iint_{\R^3\times\R^3} \frac{\rho(x)\, \rho(x')}{|x-x'|} \,dx\,dx' \,.
$$
Note the minus sign in \eqref{eq:func}; the induced interaction is attractive.

The Pekar energy is
\begin{equation}
 \label{eq:energy}
e_U = \inf \{ \mathcal E_U[\psi]:\ \psi\in H^1(\R^6)\,,\, \|\psi\|=1 \} \,,
\end{equation}
with $\|\psi\|$ denoting the $L^2$ norm of $\psi$ and $H^1$ denoting the Sobolev space, i.e., square-integrable functions whose gradient is also square-integrable. This energy is more than just an approximation, for it is asymptotically exact as $\alpha$ and $U$ tend to infinity with $U/\alpha$ fixed. This is stated in \cite{MiSp}, following the technique of \cite{LiTh}; see also \cite{DoVa}.

A benchmark for the bipolaron problem is the energy of a \emph{single polaron}. It is defined in the manner of equation \eqref{eq:func}
$$
\mathcal E[\psi] = \int_{\R^3} |\nabla \psi|^2 \,dx - 2\alpha D[|\psi|^2,|\psi|^2] \,,
$$
and
\begin{equation}
 \label{eq:energysingle}
e = \inf\{ \mathcal E[\psi]: \ \psi\in H^1(\R^3)\,,\, \|\psi\| = 1 \} \,.
\end{equation}
It is known that there is a minimizing $\psi$ for this single polaron problem and that it is unique, up to translations in $\R^3$ and multiplication by a constant phase \cite{Li}. Since it is unique, it is a radial function.

By scaling we can always reduce to the case $\alpha=1/2$ and we shall do so henceforth.

There is a considerable literature on the subject of rotation invariance of the bipolaron energy minimizer, usually formulated in the language of `one-center bipolaron versus two-center bipolaron'. The analyses are all based on variational calculations. While there seems to be general agreement that the one-center bipolaron has the lower energy, it is not completely clear that a more sophisticated variational treatment will preserve rotational symmetry, especially near the value of $U$ where the bipolaron ceases to be bound. We have shown rigorously that there is such a critical $U_c$ \cite{FLST}. Numerical variational data seems to indicate that the critical $U_c$ is rather close to $1$, namely, $U_c\sim 1.15$. This tells us that binding is a delicate matter and, indeed, the existence of such a $U_c$ was an open question for some time.

 A minimizer also exists for a bipolaron provided the energy is below the energy of twice the single polaron energy \cite{Lw}. Interestingly, the bipolaron has a minimizer with finite radius at the critical value $U=U_c$ \cite{FLS}. (The same holds for a helium atom for the critical value of the nuclear charge.) The existence of a minimizer implies that the translation invariance of the Pekar minimization problem is broken. It is, therefore, not out of the question that the minimizer for two polarons might break rotational symmetry as well in order to lessen the Coulomb repulsion. 

The value of $U$ determined by physical electrostatic considerations is always $U\geq 1$. Nevertheless, one can consider the mathematical question for small, but positive $U$ and ask whether there is a possible lack of rotational invariance in that case. After all, a rotating object like the earth becomes oblate even for the smallest amount of rotation. 

In this paper we will prove that there is no breaking of rotational symmetry for small $U$. Our strategy for proving the lack of symmetry-breaking for small $U$ is based on the following consideration. For any small $U$ the minimization problem, restricted to rotation invariant functions, has an energy minimizer, as we shall prove. The question is whether there is better minimizer without rotation invariance. Since $U$ is small, both minimizers would have to be very close to the unique (up to translations) $U=0$ minimizer, and so we can discuss the existence of a symmetry breaking minimizer by means of rigorously controlled perturbation theory. It is evident that the benefit of symmetry breaking to the repulsive energy will be proportional to $-\delta^2$, where $\delta$ measures the non-sphericity, but this contribution is multiplied by $U$. On the other hand, the increase in the rest of the energy is presumably also of the form $A\delta^2$, where $A\geq 0$ is some $U$ independent number. If $A>0$ then $\delta$ wants to be zero for small $U$.

The problem with this argument is that $A$ could be zero, in which case the $-U\delta^2$ energy would always win, no matter how small $U$ is. Most of what we do in the paper, from the mathematical point of view, is to show rigorously that $A$ is not zero, and thus there is no distortion for small $U$.

Even utilizing the result that $A>0$ for the single polaron \cite{Le}, the proof given here for the bipolaron will not be a short one. One of the complexities faced in this proof is the fact that a simple translation is a distortion that costs no energy, i.e., there are zero modes. While these are physically trivial distortions it is not a trivial matter to separate their contribution, mathematically, from the relevant ones.

\begin{theorem}\label{main}
 There is a $U_s>0$ such that for all $U< U_s$ the minimizer of $\mathcal E_U$ is unique up to translations and multiplication by a constant phase. In particular, after a translation it is rotation invariant, that is, $\psi(\mathcal Rx,\mathcal Ry)=\psi(x,y)$ for any $x,y\in\R^3$ and any $\mathcal R\in O(3)$.
\end{theorem}

It remains an \emph{open problem} to decide whether the ground state ceases to be rotation invariant for $U$ close to the critical value $U_c$.

The rest of this paper has two parts. In Part A, we reduce the proof of Theorem \ref{main} to a problem in second-order perturbation theory, namely the question of the distortion coefficient $A$. In Part B, we complete the proof of Theorem \ref{main} by showing that $A$ is positive and that zero modes play no important role.

\subsection*{Acknowledgments} We are grateful to Herbert Spohn for making us aware of this problem. Partial financial support from the U.S.~National Science Foundation through grants PHY-1068285 (R.F.), PHY-0965859 (E.L.) and the NSERC (R.S.) is acknowledged.


\section{Proof of Theorem \ref{main}. Part A}\label{sec:parta}

\subsection{Some preparations}\label{sec:prepar}

\emph{Step 1.}
In searching for the minimum in \eqref{eq:energy} we can confine our attention to non-negative, symmetric (i.e., $\psi(x,y)=\psi(y,x)$) functions. The reason is that we can replace any $\psi$ by
\begin{equation}
 \label{eq:symmnonneg}
\tilde\psi(x,y)= \sqrt{\frac12(|\psi(x,y)|^2 + |\psi(y,x)|^2)} \,.
\end{equation}
The potential energy terms remain the same and the kinetic energy term does not increase \cite[Thm. 7.8]{LiLo}.

\emph{Step 2.}
An important step is to reformulate the energy minimization problem in the following way: Define 
\begin{align}
\label{eq:energypot}
 \mathcal E_U[\psi,\Phi] & := \iint_{\R^3\times\R^3} \left( |\nabla_x\psi|^2 +|\nabla_y\psi|^2 -\Phi(x)|\psi|^2 -\Phi(y)|\psi|^2 +  \frac{U}{|x-y|} |\psi|^2 \right) \,dx\,dy \notag \notag \\
& \qquad + \frac1{8\pi} \int_{\R^3} |\nabla\Phi|^2 \,dx \iint_{\R^3\times\R^3} |\psi|^2 \,dx\,dy
\end{align}
for $\psi\in H^1(\R^6)$ and $\Phi\in\dot H^1(\R^3)$. Then
\begin{align}
\label{eq:delin}
\mathcal E_U[\psi,\Phi] \geq \|\psi\|^2 \mathcal E_U[\|\psi\|^{-1} \psi]
\end{align}
with equality if and only if $\Phi=|x|^{-1}* \rho_{\psi/\|\psi\|}$.

One advantages of this reformulation is that one can see immediately that there is a unique (up to translations) minimizer for the $U=0$ problem. For a given $\Phi$ we have, in this case, a Schr\"odinger minimization problem for two independent particles, so the best $\psi$ is a product, $\psi(x,y)=f(x)f(y)$, in which case $\rho_\psi= 2|f|^2$ and the optimal $\Phi$ is $2|x|^{-1} * |f|^2$. The problem therefore becomes
$$
e_0 = \inf\left\{ 2 \int_{\R^3} |\nabla f|^2 \,dx - 4 D[|f|^2,|f|^2] :\ \|f\|=1 \right\} \,,
$$
which coincides, up to a rescaling, with problem \eqref{eq:energysingle}. We conclude that $e_0= 8e$ and that the minimizing $f$ is unique up to translations and multiplication by a constant phase \cite{Li}.

Another advantage of this reformulation is that one sees that for any $U\geq 0$ an optimizer $\psi$, if it exists, is the ground state of a two-body Schr\"odinger operator. Therefore, $\psi$ is a multiple of a \emph{strictly positive} function \cite[Thm. 9.10]{LiLo}. This, in turn, implies that $\psi$ is symmetric (i.e., $\psi(x,y)=\psi(y,x)$), because otherwise replacing $\psi$ by \eqref{eq:symmnonneg} would strictly lower the kinetic energy \cite[Thm. 7.8]{LiLo}.

\emph{Step 3.}
In addition to the \emph{global} minimization problem \eqref{eq:energy} one can define the rotationally symmetric minimization problem, that is,
\begin{equation}
 \label{eq:energysymm}
e_U^{symm} = \inf \left\{ \mathcal E_U[\psi]:\, \psi\in H^1(\R^6)\,, \|\psi\|=1,\, \psi\ \text{rotation invariant with respect to}\ 0 \right\} .
\end{equation}
N.B.: From now on `symmetry' refers to `rotation symmetry' and not to symmetry in $x$ and $y$. We recall that the rotation invariance of $\psi$ means that $\psi(\mathcal Rx,\mathcal Ry)=\psi(x,y)$ for any $x,y\in\R^3$ and any $\mathcal R\in O(3)$. The density $\rho_\psi$ of such $\psi$ is, of course, radial and, by Newton's theorem, its potential $\rho_\psi*|x|^{-1}$ is a symmetric decreasing function bounded by $2|x|^{-1}$.

Of course, $e_U\leq e_U^{symm}$, and our goal in this paper is to investigate whether equality holds. We collect some properties of these energies. Both $e_U$ and $e_U^{symm}$ are non-decreasing, concave functions of $U$ (as infima of non-decreasing, linear functions). Because of Step 2, $e_0=e_0^{symm}$ for $U=0$. Moreover, simple trial function arguments show that $e_U \leq 2e$ and $e_U^{symm}\leq e$ for all $U$. Lewin \cite{Lw} has shown that the infimum $e_U$ is attained provided $e_U<2e$. In the appendix of this paper we shall prove an analogous result for the rotation invariant problem, with a different condition, however, namely, $e_U^{symm}<e$.

\begin{proposition}\label{exopt}
 If $e_U^{symm}<e$, then the infimum in \eqref{eq:energysymm} is attained.
\end{proposition}

The reason for the discrepancy between $2e$ for $e_U$ and $e$ for $e_U^{symm}$ is that spherical symmetry prevents the formation of two more or less separate polarons. In other words, the second polaron density has to be far away from the first, and in a radial shell, which makes it impossible to retain an energy $2e$ with an essentially unbound pair of polarons, both of which are spherically symmetric with respect to a common center.

Proposition \ref{exopt} shows the fact that for some values of $U$ the rotation invariant minimizer is \emph{not} the true minimizer. A rotation invariant minimizer is necessarily a critical point of the Pekar--Tomasevich functional. If it is a `false' minimum (i.e., its energy is bigger than the true ground state energy), its existence can possibly lead to computational difficulties for the true minimizer.

In the appendix we show two things about the rotation invariant minimization problem:
(1) There is a critical constant $U_c^{symm}$ with $1<U_c^{symm}\leq 4$ such that $e_U^{symm}<e$ for $U< U_c^{symm}$ and $e_U^{symm}=e$ for $U\geq U_c^{symm}$; see Proposition \ref{us}.\\
(2) The rotation invariant minimizer, if it exists, is a function of $|x|$, $|y|$ and $t=x\cdot y/|x| |y|$. We will show that for fixed $|x|$ and $|y|$ the minimizer is non-increasing as a function of $t$. That is, the two particles try to avoid each other; see Proposition \ref{rearr}.


\subsection{Beginning of the proof of Theorem \ref{main}}

After these preparations we are now ready to give the first part of the proof of Theorem \ref{main}. Our strategy is as follows: 
Let $U_n>0$ be a sequence such that $U_n\to 0$. For every sufficiently large $n$ there is a global minimizer $\psi_n$ and a rotation invariant minimizer $\phi_n$ corresponding to \eqref{eq:energy} and \eqref{eq:energysymm} with $U=U_n$. We shall prove that for all large $n$, $\phi_n$ and $\psi_n$ coincide up to a translation and a constant phase.

This clearly implies the theorem. Indeed, if the theorem were not true we could find a sequence $U_n$ tending to zero, and associated minimizers $\psi_n$ and $\tilde\psi_n$ which are not translates or multiples of each other. By what we are going to show, however, they are both translates and multiples of a rotation invariant minimizer $\phi_n$ if $n$ is large, which is a contradiction.

Thus, from now on we fix a sequence $U_n>0$ tending to zero. We will only consider $n$ sufficiently large such that there are a global minimizer $\psi_n$ and a rotation invariant minimizer $\phi_n$ corresponding to \eqref{eq:energy} and \eqref{eq:energysymm} with $U=U_n$. By Steps 1 and 2 above we may assume that $\psi_n$ and $\phi_n$ are positive and permutation symmetric (i.e., $\psi_n(x,y)=\psi_n(y,x)$ and similarly for $\phi_n$). Since the global minimization problem is translation invariant, we may translate $\psi_n$ in such a way that 
\begin{equation}
 \label{eq:overlap}
\alpha_n := \iint_{\R^3\times\R^3} \psi_n(x,y)\phi_n(x,y)\,dx\,dy 
= \max_{a\in\R^3} \iint_{\R^3\times\R^3} \psi_n(x-a,y-a)\phi_n(x,y)\,dx\,dy \,.
\end{equation}
With this normalization, our goal is to show that $\psi_n=\phi_n$ for all large $n$.

We decompose
$$
\psi_n = \alpha_n \phi_n + j_n \,,
$$
where, according to the definition of $\alpha_n$ and to the maximizing property in \eqref{eq:overlap},
\begin{equation}
 \label{eq:ortho}
\iint_{\R^3\times\R^3} j_n \phi_n \,dx\,dy = 0
\quad\text{and}\quad
\iint_{\R^3\times\R^3} j_n \, \vec e\cdot(\nabla_x+\nabla_y)\phi_n \,dx\,dy = 0
\ \text{for any}\ \vec e\in\Sph^2 \,.
\end{equation}
The derivative term comes from differentiating the last double integral in \eqref{eq:overlap} with respect to $a$.

We next claim that
\begin{equation}
 \label{eq:j}
\| j_n \|_{H^1} \to 0
\quad\text{as}\ n\to\infty \,.
\end{equation}
Indeed, both $\psi_n$ and $\phi_n$ are minimizing sequences for the $e_0$ problem. Since the minimizer for this problem is unique up to translations (see Step 2 above), the results of \cite{Lw} imply that
$$
\psi_n(\cdot - a_n, \cdot- a_n) \to f\otimes f
\quad\text{and}\quad
\phi_n(\cdot - b_n, \cdot- b_n) \to f\otimes f
\quad\text{in}\ H^1(\R^6)
$$
for some sequences $a_n$ and $b_n$ in $\R^3$. Here $f\otimes f$ is the $e_0$ minimizer with $f$ chosen spherically symmetric about the same origin that we fixed to formulate the problem \eqref{eq:energysymm}. Since the density of $\phi_n$ is spherically symmetric, one easily concludes that $b_n\to 0$ and, therefore, $\phi_n \to f\otimes f$ in $H^1$ as $n\to\infty$. This, together with the maximizing property of $\alpha_n$ in \eqref{eq:overlap}, implies also that $a_n\to 0$ and $\psi_n \to f\otimes f$ in $H^1$ as $n\to\infty$. Thus $\alpha_n\to 1$ and $j_n = \psi_n -\alpha_n \phi_n \to 0$ in $H^1$, as claimed in \eqref{eq:j}.

We now expand the energy of $\psi_n$ to second order in $\|j_n\|_{H^1}$. A simple but tedious computation shows that
\begin{equation}
 \label{eq:expansion}
\mathcal E_{U_n}[\psi_n] = \mathcal E_{U_n}[\phi_n] + (j_n, L_n j_n) + \mathcal O(\|j_n\|_{H^1}^3)
\end{equation}
with the linear operator
$$
L_n = \left( -\Delta_x - \Delta_y + \frac{U_n}{|x-y|} - \rho_{\phi_n} * \frac1{|x|} - \rho_{\phi_n} * \frac1{|y|} -\mu_n\right) - 4 X_{\phi_n} \,.
$$
Here $\rho_{\phi_n} * |x|^{-1}$ is an abbreviation for $\left(\rho_{\phi_n} * |\cdot|^{-1}\right)(x)$, we introduced  $\mu_n = \mathcal E_{U_n}[\phi_n] - D[\rho_{\phi_n},\rho_{\phi_n}]$ and $X_{\phi_n}$ is the integral operator on $L^2(\R^6)$ with the integral kernel
$$
X_{\phi_n}(x,y,x',y') = \frac{\phi_n(x,y)\phi_n(x',y')}{|x-x'|} + \frac{\phi_n(x,y)\phi_n(x',y')}{|y-y'|} \,.
$$
In order to see that there is no linear term in $j_n$ in the computation, we used the Euler--Lagrange equation for $\phi_n$, that is,
\begin{equation}
 \label{eq:el}
\left( -\Delta_x - \Delta_y + \frac{U_n}{|x-y|} - \rho_{\phi_n} * \frac1{|x|} - \rho_{\phi_n} * \frac1{|y|} \right)\phi_n = \mu_n \phi_n \,.
\end{equation}
(From the minimizing property of $\phi_n$ one obtains this equation only when integrated against rotation invariant functions but, since both sides are rotation invariant functions, it is true even when integrated against \emph{any} function.) We also used the fact that 
$$
\alpha_n = \left( 1- \|j_n\|^2 \right)^{1/2} = 1 - \frac12 \|j_n\|^2 + \mathcal O(\|j_n\|^4) \,.
$$
In deriving \eqref{eq:expansion} we applied standard estimates to bound all terms of higher than quadratic order by a constant times $\|j_n\|_{H^1}^3$.

The way forward is now clear: The quantity $A$ referred to in the introduction can be identified as $(j, L_n j)/\|j\|^2$ for $j$'s satisfying \eqref{eq:ortho}, and we need to prove that this is non-zero. We will prove that there are constants $N\geq 1$ and $c>0$ such that $(j, L_n j) \geq c\|j\|^2$ for all $n\geq N$ and all $j$ satisfying the orthogonality conditions \eqref{eq:ortho}. Then \eqref{eq:expansion} will imply that
$$
\mathcal E_{U_n}[\psi_n] \geq \mathcal E_{U_n}[\phi_n] + c \|j_n\|^2 + \mathcal O(\|j_n\|_{H^1}^3) \,.
$$
Since $\psi_n$ is, by assumption, an energy minimizer, we have necessarily $j_n=0$ for all large $n$, which means $\psi_n=\phi_n$, as we intended to prove.

Thus we are left with proving a lower bound on $(j, L_n j)$. We show this perturbatively by analyzing the $U=0$ case.


\subsection{The Hessian}

Instead of working directly with the operator $L_n$, it is more convenient to work with a closely related operator, namely the Hessian $H_n$ of $\mathcal E_{U_n}$ at $\phi_n$. That is, for any normalized, real-valued $j \in H^1(\R^6)$ we define
\begin{equation}
 \label{eq:hessian}
\left.\frac{d^2}{d\epsilon^2}\right| _{\epsilon=0} \mathcal E_{U_n} \left[ \frac{\phi_n+\epsilon j}{\sqrt{1+ 2 \epsilon (\phi_n,j) + \epsilon^2}} \right] = (j, H_n j) \,.
\end{equation}
A similar computation as before shows that
$$
H_n = L_n +  |k_n \rangle \langle\phi_n| + |\phi_n\rangle\langle k_n| + \beta_n |\phi_n\rangle\langle\phi_n| \,,
$$
where
$$
\beta_n = 4 \left( \left( \phi_n,\left( -\Delta_x - \Delta_y + \frac{U_n}{|x-y|}\right)\phi_n \right) - 3D[\rho_{\phi_n},\rho_{\phi_n}] \right)
$$
and
$$
k_n = -2 \left( \left( -\Delta_x - \Delta_y + \frac{U_n}{|x-y|}\right)\phi_n - 2 \rho_{\phi_n} * |x|^{-1} \phi_n - 2 \rho_{\phi_n} * |y|^{-1} \phi_n \right) \,.
$$
The expressions for $\beta_n$ and $k_n$ can be somewhat simplified using equation \eqref{eq:el} for $\phi_n$, but we will not need this. The only thing that is relevant for us is that $(j,H_n j) = (j,L_n j)$ if $j$ satisfies the first orthogonality condition \eqref{eq:ortho}.

We collect two facts about the operators $H_n$. First, $H_n$ (which commutes with angular momentum since $\phi_n$ is rotation invariant) is non-negative on the subspace of angular momentum zero. This follows from the minimizing property of $\phi_n$. Second, the functions $\phi_n$ and $\vec e\cdot(\nabla_x+\nabla_y)\phi_n$, $e\in\Sph^2$, are in the kernel of $H_n$. For $\phi_n$, this follows immediately from the Euler--Lagrange equation \eqref{eq:el} for $\phi_n$, and for $\vec e\cdot(\nabla_x+\nabla_y)\phi_n$, this follows by differentiating the equation \eqref{eq:el} for $\phi_n(x+t\vec e,y+t\vec e)$ (which is the same as that for $\phi_n$) with respect to $t$ at $t=0$.

The following proposition says that, for small $U_n$, $H_n$ is strictly positive away from the zero modes found above.

\begin{proposition}\label{pos}
There is a number $c>0$ and an $N\geq1$ such that for all $n\geq N$ and for all $j$ which are orthogonal to $\phi_n$ and to $\vec e\cdot(\nabla_x+\nabla_y)\phi_n$, $\vec e\in\Sph^2$, the following estimate holds:
$$
(j,H_n j) \geq c\|j\|^2 \,.
$$
Equivalently, if $P_n$ denotes the projection onto the four-dimensional space spanned by $\phi_n$ and $\vec e\cdot(\nabla_x+\nabla_y)\phi_n$, $e\in\Sph^2$, then for $n\geq N$
$$
P_n^\bot H_n P_n\bot \geq c P_n^\bot \,.
$$
\end{proposition}

Note that, as remarked above, $(j,H_n j)=(j,L_n j)$ if $(\phi_n,j) = 0$ and, therefore, Proposition \ref{pos} concludes the proof of Theorem \ref{main}.\qed

\medskip

Since $\phi_n\to f\otimes f$ in $H^1(\R^6)$ as $n\to\infty$, the operators $H_n$ tend to the corresponding operator at $U=0$ in norm resolvent sense. In particular, the eigenvalues converge and, therefore, it suffices to prove Proposition \ref{pos} in the case $U=0$. This is the topic of the next section.

The reason that we have chosen to work with the operator $H_n$ instead of $L_n$ is that the operator $L_n$ has a negative eigenvalue. (This follows from the variational principle since $(\phi_n,L_n\phi_n) = -4 (\phi_n X_{\phi_n},\phi_n)<0$ by \eqref{eq:el}.) The positivity of $(j,L_n j)$ asserted in Proposition \ref{pos} therefore crucially relies on the orthogonality condition \eqref{eq:ortho}. This condition is not easy to use, however, since $\phi_n$ is not an eigenfunction of $L_n$. In contrast, it is an eigenfunction of $H_n$. We also note that $H_n$ is the operator that automatically takes care of the normalization condition, without any reference to orthogonality, and, therefore, is most directly connected to the coefficient $A$ mentioned in the introduction.


\section{Proof of Theorem \ref{main}. Part B}

As explained in Step 2 of the previous section, at $U=0$ we have $\phi(x,y)=f(x)f(y)$, where $f$ is a radial decreasing function on $\R^3$ with $\int f^2\,dx =1$ and
\begin{equation}
 \label{eq:onepart}
h f = 0\,,\qquad h=-\Delta - 2f^2 *\frac{1}{|x|} - \frac\mu2 \,.
\end{equation}
Here, $\mu=e_0-4D[f^2,f^2]=8e-4D[f^2,f^2]$. The Hessian of $\mathcal E_0$ at $f\otimes f$ is the operator on $L^2(\R^6)$ given by
$$
H = L +  |r\otimes f+ f\otimes r \rangle \langle f\otimes f| + |f\otimes f\rangle\langle r\otimes f+f \otimes r| + \beta |f\otimes f\rangle\langle f\otimes f| \,,
$$
where
$$
L = h_x + h_y - 4 X_{f\otimes f}
$$
and where $\beta$ is a constant and $r$ is a function in $L^2(\R^3)$. More precisely,
$$
\beta = 8 \left( \left( f,-\Delta f \right) - 6 D[f^2,f^2] \right)
$$
and
$$
r=-2 \left( -\Delta - 4 f^2 * |x|^{-1}\right) f \,.
$$
The main result of this section is the following $U=0$ analogue of Proposition \ref{pos}. We note that the three-dimensional space spanned by $\vec e\cdot(\nabla_x+\nabla_y)\phi$, $\vec e\in\Sph^2$, coincides with the space spanned by $f\otimes f' Y_{1,m}+ f'Y_{1,m}\otimes f$, $m=-1,0,1$, where $Y_{1,m}$ are spherical harmonics of degree one.

\begin{proposition}\label{pos0}
There is a number $C>0$ such that for all $j\in H^1(\R^6)$ with $j(x,y)=j(y,x)$ that satisfy
$$
(f\otimes f,j) = (f\otimes f' Y_{1,m}+ f'Y_{1,m}\otimes f, j) = 0
\qquad\text{for all}\ m=-1,0,1,
$$
the following estimate holds:
$$
(j,H j) \geq C\|j\|^2 \,.
$$
Equivalently, if $P$ denotes the projection onto the four-dimensional space spanned by $f\otimes f$ and $f\otimes f' Y_{1,m}+ f'Y_{1,m}\otimes f$, $m=-1,0,1$, then
$$
P^\bot H P^\bot \geq C P^\bot \,.
$$
\end{proposition}

We shall deduce this from a result about the Hessian of the one-particle functional, which we discuss in the following subsection.


\subsection{The one-particle Hessian}

We recall that the function $f$ minimizes the one-polaron functional
\begin{equation}
 \label{eq:onepol}
\int_{\R^3} |\nabla \psi|^2 \,dx - 2 D[|\psi|^2,|\psi|^2]
\end{equation}
and that the corresponding Euler--Lagrange equation reads $hf=0$ with $h$ from \eqref{eq:onepart}. Moreover, the Hessian of the above one-polaron functional at $f$ reads
$$
H^{(1)} = L^{(1)} + |r\rangle\langle f| + |f\rangle\langle r| + \gamma |f\rangle\langle f|
$$
with
$$
L^{(1)} = h - 4 x_f \,.
$$
Here $x_f$ is the operator on $L^2(\R^3)$ with the integral kernel $f(x)|x-x'|^{-1} f(x')$ and $r$ is the same function as in the expression of the two-particle operator $H$. The following theorem is the one-particle analogue of Proposition \ref{pos0}. Its proof relies heavily on previous work of Lenzmann \cite{Le}.

\begin{proposition}\label{pos1}
There is a number $C'>0$ such that for all $j\in H^1(\R^3)$ that satisfy
$$
(f,j) = (f' Y_{1,m}, j) = 0
\qquad\text{for all}\ m=-1,0,1,
$$
the following estimate holds:
$$
(j, H^{(1)} j) \geq C' \|j\|^2 \,.
$$
\end{proposition}

Similarly as in the discussion before Proposition \ref{pos}, we note that $f$ and $f' Y_{1,m}$ are zero modes of $H^{(1)}$. Thus the proposition says that there are no other zero modes.

\begin{proof}
 We first note that $H^{(1)}$ is non-negative, since $f$ is a minimizer of the one-polaron problem. We now argue that $L^{(1)}$ has exactly one negative eigenvalue. Indeed, since $(j,L^{(1)}j)=(j,H^{(1)}j)$ if $(f,j)=0$ and since $H^{(1)}\geq 0$, the variational principle implies that $L^{(1)}$ has at most one negative eigenvalue. On the other hand, the Euler--Lagrange equation for $f$ implies that $(f,L^{(1)}f)= -4 (f,x_f f)<0$, which, again by the variational principle, means that $L^{(1)}$ has at least one negative eigenvalue. This proves the claim.

We next recall Lenzmann's result \cite{Le}, which states that
\begin{equation}
 \label{eq:kernel}
\ker L^{(1)} = \spa\{ f' Y_{1,m}:\ m=-1,0,1 \} \,.
\end{equation}
We also note that $L^{(1)}$ commutes with angular momentum. Moreover, its essential spectrum starts at $-\mu/2>0$ and, therefore, \eqref{eq:kernel} implies that there is a constant $c'>0$ such that
$$
(j, H^{(1)} j) = (j, L^{(1)} j) \geq c' \|j\|^2
$$
for all $j$ with angular momentum $l\geq 2$ and for all $j$ with angular momentum $l=1$ satisfying the additional constraint that $(f' Y_{1,m}, j) = 0$ for all $m=-1,0,1$. (Actually, the argument of \cite{Le} shows that the best constant $c'$ is achieved either for $l=1$ or for $l=2$.)

Since the function $r$ appearing in the definition of $H^{(1)}$ is radial, the operator $H^{(1)}$ commutes with angular momentum. Therefore, the previous discussion reduces the proof of Proposition \ref{pos1} to finding a lower bound on $(j, H^{(1)} j)$ for radial $j$ satisfying $(f,j)=0$. In other words, we have to exclude the possibility that $0$ is a degenerate eigenvalue of $H^{(1)}$ restricted to $l=0$.

As an aside, before completing the proof, we show that there is a radial function $R$ such that
$$
 L^{(1)} R = \mu f
$$
and
$$
(R,f) = \frac12 \,.
$$
This is also contained in \cite{Le}, but we include the short proof for the convenience of the reader. We define $f_\beta(x)=\beta^2 f(\beta x)$ and note that the Euler--Lagrange equation for $f$ implies the following equation for $f_\beta$,
$$
\left( -\Delta - 2 f_\beta^2 * \frac1{|x|} \right) f_\beta = \beta^2 \frac\mu2 f_\beta \,.
$$
By differentiating this equation at $\beta=1$ we obtain $L^{(1)} \frac{d}{d\beta}|_{\beta=1}f_\beta = \mu f$ and thus
$$
R= \left. \frac{d}{d\beta}\right|_{\beta=1} f_\beta = 2f + x\cdot\nabla f
$$
satisfies the desired equation and is obviously radial. The overlap $(R,f)$ is computed by differentiating the identity $\int f_\beta(x)^2 \,dx = \beta$ at $\beta=1$.

Having found the function $R$ we now conclude the proof by showing that $0$ is a simple eigenvalue of $ H^{(1)}$ restricted to radial functions. Thus, let $v$ be a radial function satisfying $ H^{(1)} v =0$ and $(f,v)=0$. We need to show that $v=0$. By the expression for $ H^{(1)}$,
$$
 L^{(1)} v + (r,v) f = 0 \,.
$$
We define $\tilde v= v + \mu^{-1} (r,v) R$ with $R$ as constructed above. Thus $ L^{(1)} \tilde v = 0$. Since $\tilde v$ is radial, Lenzmann's result \eqref{eq:kernel} implies that $\tilde v=0$. From the orthogonality condition of $v$ we infer that 
$$
0=(f,\tilde v) = \mu^{-1} (r,v) (f,R) \,.
$$
Thus, since $(f,R)\neq 0$, we have $(r,v)=0$, and therefore $v = \tilde v = 0$, as claimed. This concludes the proof of Proposition \ref{pos1}.
\end{proof}


\subsection{Proof of Proposition \ref{pos0}}

We consider the operator $H$ in the subspace $L^2_{symm}(\R^6)$ of all functions $j\in L^2(\R^6)$ satisfying $j(x,y)=j(y,x)$. We decompose this space as $L^2_{symm}(\R^6) = \mathcal H_0 \oplus \mathcal H_1$, where
$$
\mathcal H_0 = \{ \alpha f\otimes f + f\otimes g + g\otimes f:\ \alpha\in\C,\, (f,g)=0 \}
$$
and
$$
\mathcal H_1 = \spa\{f\}^\bot \otimes \spa\{f\}^\bot \,.
$$
Since $hf =0$ and since the rank one operators entering in the definition of $H$ only involve terms of the form $\tilde k \otimes f + f \otimes \tilde k$ (with $\tilde k$ being either $r$ or $f$), the operator $H$ leaves both subspaces invariant and we can study it separately on each subspace.

We observe that $ H$ coincides with $ L$ on $\mathcal H_1$ and that, moreover, the operator $X_{f\otimes f}$ vanishes on that space. Therefore, on $\mathcal H_1$ the operator $ L$ is just a sum of two one-body operators $h$. Since $f$ is positive, it is the ground state of $h$ and since the essential spectrum of $h$ starts at $-\mu/2>0$, $h$ has a gap $\delta>0$ above zero. We conclude that $L \geq 2\delta$ on $\mathcal H_1$.

We now turn to the space $\mathcal H_0$. More precisely, we are only interested in the space
$$
\mathcal H_0 \cap \ran P^\bot = \{ f\otimes g + g\otimes f:\ (f,g)=(f' Y_{1,m},g)=0 \} \,.
$$
For $j=f\otimes g+ g\otimes f$ from this space we have
$$
(j,Lj) = 2 (g,(h-4x_f)g) \,,
$$
and by Proposition \ref{pos1}, this is bounded from below by $2 C' \|g\|^2 = C' \|j\|^2$. This completes the proof of Proposition \ref{pos0}.
\qed



\appendix

\section{Existence of a rotation invariant optimizer}

Our goal in this section is to prove Proposition \ref{exopt}, that is, the existence of a rotation invariant optimizer. The heart of the proof is the following lemma.

\begin{lemma}\label{onee}
 Assume that $\psi_n\in H^1(\R^6)$ is a sequence of normalized, rotation invariant functions such that $\sqrt{\rho_{\psi_n}}$ has a weak limit $\sqrt\rho$ in $H^1(\R^3)$. Then
$$
\liminf_{n\to\infty} \mathcal E_0[\psi_n] \geq e \left( \int_{\R^3} \rho \,dx \right)^3 \,.
$$
\end{lemma}

\begin{proof}
  As in Step 2 in Subsection \ref{sec:prepar} we rewrite the energy in terms of the potential $\Phi_n := \rho_n *|x|^{-1}$ generated by $\rho_n := \rho_{\psi_n}$. By the Hoffmann--Ostenhof inequality \cite{HoHo} we obtain
\begin{align}\label{eq:enlower}
 \mathcal E_0[\psi_n] & = \iint_{\R^3\times\R^3} \left( |\nabla_x\psi_n|^2 +|\nabla_y\psi_n|^2 -\Phi_n(x)|\psi_n|^2 -\Phi_n(y)|\psi_n|^2 \right) \,dx\,dy \notag \\
& \qquad + \frac1{8\pi} \int_{\R^3} |\nabla\Phi_n|^2 \,dx \notag \\
& \geq \int_{\R^3} \left( |\nabla \sqrt{\rho_n}|^2 - \Phi_n \rho_n \right) \,dx 
+ \frac1{8\pi} \int_{\R^3} |\nabla\Phi_n|^2 \,dx \,.
\end{align}
Of course, we may pass to a subsequence and assume that $\mathcal E_0[\psi_n]$ has a finite limit. Thus we infer that the sequence $\Phi_n$ is bounded in $\dot H^1(\R^3)$ and, after passing to another subsequence if necessary, we may assume it has a weak limit $\Phi$ in $\dot H^1(\R^3)$. Below we shall argue that
\begin{align}
 \label{eq:weakstrong}
\lim_{n\to\infty} \int_{\R^3} \left( \Phi_n \rho_n -\Phi \rho \right) \,dx = 0 \,.
\end{align}
Once this is proved, we conclude from \eqref{eq:enlower} and the lower semi-continuity of the terms $\int |\nabla\sqrt{\rho_n}|^2 \,dx$ and $\int |\nabla\Phi_n|^2 \,dx$ that
$$
\liminf_{n\to\infty} \mathcal E_0[\psi_n] 
\geq \int_{\R^3} \left( |\nabla \sqrt{\rho}|^2 - \Phi \rho \right) \,dx + \frac1{8\pi} \int_{\R^3} |\nabla\Phi|^2 \,dx \,.
$$
By the same argument as in \eqref{eq:delin} and by the definition of the single-polaron energy we find that
$$
\int_{\R^3} \left( |\nabla \sqrt{\rho}|^2 - \Phi \rho \right) \,dx + \frac1{8\pi} \int_{\R^3} |\nabla\Phi|^2 \,dx
\geq \int_{\R^3} |\nabla \sqrt{\rho}|^2 \,dx - D[\rho,\rho] \geq e \left( \int_{\R^3} \rho \,dx \right)^3 \,.
$$
The last inequality follows by scaling $\rho$. This is the claimed lower bound.

We are left with proving \eqref{eq:weakstrong}. First, by Sobolev embeddings $\Phi_n$ converges to $\Phi$ weakly in $L^6$ and we have $\rho\in L^1\cap L^3\subset L^{6/5}$. Thus, we only need to prove
$$
\lim_{n\to\infty} \int_{\R^3}  \Phi_n \left(\rho_n - \rho \right) \,dx = 0 \,.
$$
At this point we use the spherical symmetry of $\rho_n$ and $\Phi_n$. Newton's theorem implies that $0\leq \Phi_n(x) \leq 2|x|^{-1}$. Thus, for any $R>0$,\
$$
\left| \int_{|x|>R}  \Phi_n \left(\rho_n - \rho \right) \,dx \right| \leq 8 R^{-1} \,,
$$
which can be made arbitrarily small, uniformly in $n$, by choosing $R$ large. On the other hand,
$$
\left| \int_{|x|\leq R}  \Phi_n \left(\rho_n - \rho \right) \,dx \right| 
\leq \left( \int_{\R^3}  \Phi_n^2 \left(\sqrt{\rho_n} + \sqrt{\rho} \right)^2 \,dx \right)^{1/2} \left( \int_{|x|\leq R}  \left(\sqrt{\rho_n} - \sqrt{\rho} \right)^2 \,dx \right)^{1/2} \,.
$$
By Rellich--Kondrashov the second term on the right side tends to zero as $n\to\infty$ for every fixed $R>0$. Moreover, the first term is bounded uniformly in $n$, since
$$
\int_{\R^3}  \Phi_n^2 \left(\sqrt{\rho_n} + \sqrt{\rho} \right)^2 \,dx \leq 
\left( \int_{\R^3}  \Phi_n^6 \,dx \right)^{1/3} 
\left( \int_{\R^3}  \left(\sqrt{\rho_n} + \sqrt{\rho} \right)^3 \,dx \right)^{2/3}  \,.
$$
(The fact that the $\rho_n$-term is uniformly bounded follows from the fact that $\rho_n$ is uniformly bounded in $L^1\cap L^3$ by Sobolev inequalities.) This concludes the proof of \eqref{eq:weakstrong}, and therefore the lemma is proven.
\end{proof}

The following lemma is well known. We include the proof for the convenience of the reader. An instructive example to keep in mind is where $\psi_n(x,y) = f(x)g_n(y) + g_n(x)f(y)$ where $g_n$ and $f$ are $H^1(\R^3)$ functions with disjoint support and where $g_n$ converges weakly to zero in $H^1(\R^3)$.

\begin{lemma}\label{mass}
 Assume that $\psi_n$ is normalized and converges weakly to zero in $H^1(\R^6)$ and assume that $\sqrt{\rho_{\psi_n}}$ converges weakly in $H^1(\R^3)$ to some $\sqrt\rho$. Then $\int_{\R^3} \rho\,dx \leq 1$.
\end{lemma}

\begin{proof}
  For any $R>0$ we write
\begin{align*}
\int_{|x|<R} \rho_{\psi_n}(x) \,dx 
& = \iint_{|x|<R} |\psi_n(x,y)|^2 \,dx\,dy + \iint_{|y|<R} |\psi_n(x,y)|^2 \,dx\,dy \\
& = 1 - \iint_{|x|>R,|y|>R} |\psi_n(x,y)|^2 \,dx\,dy + \iint_{|x|<R, |y|<R} |\psi_n(x,y)|^2 \,dx\,dy \\
& \leq 1 + \iint_{|x|<R, |y|<R} |\psi_n(x,y)|^2 \,dx\,dy \,.
\end{align*}
Since $\psi_n$ converges weakly to zero in $H^1(\R^6)$, the Rellich--Kondrashov theorem implies that the last double integral on the right side tends to zero as $n\to\infty$. Thus, again by the Rellich--Kondrashov theorem now applied to $\sqrt{\rho_{\psi_n}}$,
$$
\int_{|x|<R} \rho(x) \,dx = \lim_{n\to\infty} \int_{|x|<R} \rho_{\psi_n}(x) \,dx \leq 1 \,.
$$
Since this is true for any $R>0$, we obtain the assertion.
\end{proof}

\begin{proof}[Proof of Proposition \ref{exopt}]
Let $\psi_n$ be a minimizing sequence for $e_U^{symm}$ with $e_U^{symm}<e$. Since $\mathcal E_U[\psi_n]$ is bounded, $\psi_n$ is bounded in $H^1$ and, after passing to a subsequence, it converges weakly in $H^1(\R^6)$ to a function $\psi$. We decompose $\psi_n = \psi + \tilde\psi_n$, where now $\tilde\psi_n$ converges weakly to zero in $H^1(\R^6)$, and note that
\begin{equation}
 \label{eq:split3}
1=\|\psi_n\|^2 = \|\psi\|^2 + \|\widetilde\psi_n\|^2 + o(1) \,.
\end{equation}
We claim that asymptotically, the energy of the minimizing sequence splits as follows,
\begin{equation}
 \label{eq:split}
\mathcal E_U[\psi_n] = \mathcal E_U[\psi,\Phi_n] + \mathcal E_U[\widetilde\psi_n,\Phi_n] + o(1) \,.
\end{equation}
Here $\Phi_n=\rho_{\psi_n}*|x|^{-1}$, and we recall that the notation $\mathcal E_U[\psi,\Phi]$ was introduced in \eqref{eq:energypot}. The proof of \eqref{eq:split} relies on the weak convergence in $H^1$ for the positive terms in the energy functional and on the fact that
$$
\iint_{\R^3\times\R^3} \Phi_n(x) \widetilde\psi_n(x,y) \psi(x,y) \,dx\,dy = o(1) \,.
$$
The proof of the latter relation also uses the weak convergence in $H^1$, together with the fact that $\Phi_n$ is bounded in $L^\infty$. The details are as in \cite[Eq. (5.13)]{FLS} and are omitted. 

It follows from \eqref{eq:split}, together with \eqref{eq:delin} and \eqref{eq:split3} that
\begin{equation}
 \label{eq:split2}
\mathcal E_U[\psi_n] \geq \|\psi\|^2 \mathcal E_U[\|\psi\|^{-1}\psi] + (1-\|\psi\|^2) \mathcal E_U[v_n] + o(1) \,,
\end{equation}
where $v_n = \|\widetilde\psi_n\|^{-1} \widetilde\psi_n$. Here we interpret $\mathcal E_U[\|\psi\|^{-1}\psi]$ as zero if $\psi\equiv0$, and similarly for $\mathcal E_U[v_n]$.

Given \eqref{eq:split2}, it is easy to deduce that $\|\psi\|=1$. Indeed, we argue by contradiction and assume that $\|\psi\|<1$, which is the same, by \eqref{eq:split3}, as assuming that $\|\tilde\psi_n\|$ has a non-zero limit. Hence $v_n$ converges weakly to zero in $H^1(\R^6)$. By the Hoffmann-Ostenhof inequality \cite{HoHo} the square roots of the corresponding densities $\sigma_n = \rho_{v_n}$ are bounded in $H^1(\R^3)$ and hence, after passing to a subsequence, have a weak limit $\sqrt\sigma$ in $H^1(\R^3)$. Since all $v_n$ are rotation invariant, we learn from Lemmas \ref{onee} and \ref{mass} that
$$
\liminf_{n\to\infty} \mathcal E_U[v_n] \geq e \left( \int_{\R^3} \sigma \,dx \right)^3 \geq e \,.
$$
On the other hand, we trivially have $\mathcal E_U[\|\psi\|^{-1}\psi] \geq e_U^{symm}$. (This also holds if $\psi\equiv 0$ with our convention.) Thus \eqref{eq:split2} implies that
$$
e_U^{symm} = \lim_{n\to\infty} \mathcal E_U[\psi_n] \geq \|\psi\|^2 e_U^{symm} + (1-\|\psi\|^2) e \,.
$$
Since $\|\psi\|<1$, this contradicts our assumption that $e_U^{symm}<e$.

Thus we have shown that $\|\psi\|=1$, and now \eqref{eq:split2} implies that $\mathcal E_U[\psi_n] \geq \mathcal E_U[\psi] + o(1)$, from which we deduce that $e_U^{symm} \geq \mathcal E_U[\psi]$, that is, $\psi$ is a minimizer. This completes the proof of Proposition \ref{exopt}.
\end{proof}


\section{On the rotation invariant optimization problem}

In the next two propositions we present some interesting facts about the rotation invariant minimizers. This appendix is not needed in the rest of the paper.

\begin{proposition}\label{us}
 If $U<1$, then $e_U^{symm} < e$. On the other hand, if \eqref{eq:el} has a rotation-invariant solution $0\not\equiv \phi\in H^1(\R^6)$, then $U< 4$. In particular, $e_U^{symm} = e$ for $U\geq 4$.
\end{proposition}

\begin{proof}
 The first part of the proposition follows by a simple variational computation in the manner of Zhislin's theorem; see, e.g., \cite[Thm. 12.2]{LiSe}. We write $\psi(x,y) = [f_R(x) \eta_R(y) + \eta_R(x) f_R(y)]/\sqrt 2$, where $f_R$ is the one-polaron function smoothly cut off at some large radius $R$, appropriately normalized (thereby making an error in the energy of the order $o(1/R)$) and where $\eta_R$ is a normalized, radial function with support in the shell $10R \leq |x|\leq 11 R$. Thus the total energy is, by Newton's theorem 
$$
e+o(1/R) + (U-1)\int_{\R^3} \frac{|\eta_R(y)|^2}{|y|}\,dy \leq e+o(1/R) + \frac{U-1}{11 R} \,.
$$
For sufficiently large $R$, this number is less then $e$.

The second part is an adaptation of the $N<2Z+1$ theorem in \cite[Sec. VI.a]{Li0}. In the present situation $N=2$ and the effective $Z=\int \rho_\phi \,dx/U =2/U$.
\end{proof}

Our second result in this section concerns the correlation of the particles in the optimizing rotation invariant state. We state this in form of a rearrangement inequality, where the rearrangement is defined as follows. We identify a rotation invariant function $\psi(x,y)$ with a function $u(r,s,t)$, where $r=|x|$, $s=|y|$ and $t=  x\cdot y /|x| |y|$. For fixed $r,s\geq 0$, we denote by $u^*(r,s,\cdot)$ the unique, non-increasing function on $[-1,1]$ which is equi-measurable with $|u(r,s,\cdot)|$ (in the sense of Lebesgue measure on $[-1,1]$). If $u(r,s,t)$ came from a function $\psi(x,y)$, we also use the notation $\psi^*(x,y)$ for $u^*(r,s,t)$. Then $|\psi|$ and $\psi^*$ are equi-measurable and, in particular, $\|\psi\|= \|\psi^*\|$.

\begin{proposition}\label{rearr}
 For any $U>0$ and any rotation invariant $\psi\in H^1(\R^6)$ with $\|\psi\|=1$ one has $\mathcal E_U[\psi^*] < \mathcal E_U[\psi]$ unless $\psi=\psi^*$ a.e.
\end{proposition}

In particular, if $\psi$ is an optimizer for $e^{symm}_U$, then for any fixed $|x|$ and $|y|$, $\psi(x,y)$ is a non-increasing function of $x\cdot y /|x| |y|$.

\begin{proof}
 Since $\rho_\psi=\rho_{\psi^*}$ because of equi-measurability, the $D[\rho,\rho]$ term in the energy functional does not change when $\psi$ is replaced by $\psi^*$. Moreover, the repulsion term improves (unless $\psi=\psi^*$) since, for fixed $r,s> 0$,
$$
 \int_{-1}^1 \frac{|u(r,s,t)|^2}{(r^2-2rst+ s^2)^{1/2}} \,dt 
> \int_{-1}^1 \frac{|u^*(r,s,t)|^2}{(r^2-2rst+ s^2)^{1/2}} \,dt
$$
unless $u(r,s,t)=u^*(r,s,t)$. This is a simple rearrangement inequality proved as in \cite[Thm. 3.4]{LiLo}. It uses the fact that $(r^2-2rst+ s^2)^{-1/2}$ is an increasing function of $t$.

For the kinetic energy we compute
$$
|\nabla_x\psi|^2 = \left|\frac{\partial u}{\partial r}\right|^2 + \frac1{r^2} (1-t^2) \left|\frac{\partial u}{\partial t}\right|^2
$$
and similarly for $\nabla_y\psi$, so that
\begin{align*}
& \iint_{\R^3\times\R^3} \left( |\nabla_x\psi|^2 +|\nabla_y\psi|^2 \right) \,dx\,dy \\
& \quad =8\pi^2 \int_0^\infty \int_0^\infty \int_{-1}^1 \left( \left|\frac{\partial u}{\partial r}\right|^2 + \left|\frac{\partial u}{\partial s}\right|^2 + \left( \frac1{r^2}+\frac1{s^2}\right) (1-t^2) \left|\frac{\partial u}{\partial t}\right|^2 \right) r^2 s^2 \,dt\,ds\,dr \,.
\end{align*}
We now argue that for any fixed $r$ and $s$, one has
$$
\int_{-1}^1 \left|\frac{\partial u}{\partial r}\right|^2 \,dt \geq \int_{-1}^1 \left|\frac{\partial u^*}{\partial r}\right|^2 \,dt \,,
$$
and similarly for $\frac{\partial u}{\partial s}$, and
$$
\int_{-1}^1 (1-t^2) \left|\frac{\partial u}{\partial t}\right|^2 \,dt
\geq \int_{-1}^1 (1-t^2) \left|\frac{\partial u^*}{\partial t}\right|^2 \,dt \,.
$$
Of course, these two inequalities will complete the proof of the proposition. For the proof of the first inequality we approximate $\frac{\partial u}{\partial r}(r,s,t)$ by $h^{-1}(u(r+h,s,t)-u(r,s,t))$. Then the $t$-integrals of the terms $|u(r+h,s,t)|^2$ and $|u(r,s,t)|^2$ do not change under symmetrization, whereas
$$
\re \int_{-1}^1 u(r+h,s,t)u(r,s,t) \,dt
\leq \int_{-1}^1 u^*(r+h,s,t)u^*(r,s,t)\,dt
$$
by a simple rearrangement inequality. As $h\to 0$ we obtain the first one of two claims. For the proof of the second inequality we consider the function $g$ on $\Sph^2$, given in spherical coordinates $(\phi,\theta)$ by $g(\phi,\theta) = u(r,s,\cos\theta)$. Then
$$
\int_{\Sph^2} |\nabla_\omega g|^2 \,d\omega = 2\pi \int_{-1}^1 (1-t^2) \left|\frac{\partial u}{\partial t}\right|^2 \,dt \,.
$$
Since passing from $u(r,s,\cdot)$ to $u^*(r,s,\cdot)$ corresponds to the usual symmetrization of $g$ on the sphere, the latter inequality follows from standard symmetrization results; see, e.g., \cite{Ba}. This finishes the proof of Proposition \ref{rearr}.
\end{proof}


\bibliographystyle{amsalpha}

\end{document}